\documentclass[reqno,10pt,centertags,draft]{amsart}
\usepackage{amsmath,amsthm,amscd,amssymb,latexsym,upref}
\date{\today}

\newcommand{\bbR}{{\mathbb{R}}}

\newcommand{\bbC}{{\mathbb{C}}}

\newcommand{\bbT}{{\mathbb{T}}}

\newcommand{\cN}{{\mathcal{N}}}

\newcommand{\cS}{{\mathcal{S}}}
\newcommand{\cE}{{\mathcal{E}}}

\renewcommand{\Re}{\text{\rm Re}}

\newcommand{\Mth}{\langle \phi \rangle_{I}}
\newcommand{\Ms}{\langle f \rangle_{I}}

\allowdisplaybreaks \numberwithin{equation}{section}
\newtheorem{theorem}{Theorem}[section]

\newtheorem{lemma}[theorem]{Lemma}
\newtheorem{proposition}[theorem]{Proposition}
\newtheorem{corollary}[theorem]{Corollary}

\theoremstyle{definition}

\newtheorem{remark}[theorem]{Remark}
\newtheorem{problem}[theorem]{Problem}

\begin{document}

\title[]{Remarks on Nehari's problem, matrix $A_2$ condition, and weighted bounded mean oscillation}
\author[]{ A. Volberg
and P. Yuditskii}

\address{Institute for Analysis, Johannes Kepler University Linz,
A-4040 Linz, Austria}
\email{Petro.Yudytskiy@jku.at}
\thanks{
Partially supported by NSF grant DMS-0501067 and  the Austrian Founds
FWF, project number: P20413--N18}

\address{
Department of Mathematics, Michigan State University, East Lansing,
MI 48824, USA} \email{volberg@math.msu.edu}

\date{\today}

\begin{abstract}
We consider Nehari's problem in the case of non-uniqueness of solution.
The solution set is then parametrized by the unit ball of $H^{\infty}$
by means of so-called {\em regular generators} --- bounded holomorphic functions $\phi$.
The definition of {\em regularity} is given below, but let us mention now that
1) the following assumption on modulus of $\phi$ is sufficient for {\em regularity}:
$\frac{1}{1-|\phi|^2}\in L^1(\mathbb{T})$; 2) there is no necessary and sufficient condition
of {\em regularity} on bounded holomorphic $\phi$ in terms of $|\phi|$ on $\mathbb{T}$, \cite{Kh1}.
This makes reasonable the attempt to find a weaker sufficient condition on $|\phi|$ than the condition in 1).
This is done here.
Also we are discussing  certain new necessary and sufficient conditions of {\em regularity} in terms
of bounded mean (weighted)
oscillations of $\phi$. They involve the matrix $A_2$ condition from \cite{TV}.
\end{abstract}

\maketitle

\section{Introduction}
Recent developments in the inverse scattering/spectral theory
 \cite{VYu,TTh, CK, MP, R, Sv1, Sv2} stimulated our interest
to an old question on the description of   the Nehari problem solutions
 set  (precisely the question is formulated
 in Problem \ref{pr1} below). The Nehari
problem is strongly related to the Nonlinear Fourier Analysis
 \cite{TTh}, or what is basically the same,
to the inverse scattering problem for CMV matrices \cite{KhPYu}.

Here we consider $L^p$ spaces of functions on the unit circle $\bbT$
 and their Hardy subspaces $H^p$.
Recall that the famous Nehari Theorem describes projections of
 functions of the unite ball
of $L^\infty$ onto the Hardy space $H^2_-$, see e.g. textbooks
 \cite{Nik, Garnett}.
Let $P_-$ be the Riesz projector $P_-:L^2\to H^2_-$.
The function $F_-\in H^2_-$ possesses the representation
\begin{equation}\label{repr}
    F_-=P_-f,\quad \|f\|_\infty\le 1
\end{equation}
if and only if the corresponding Hankel operator
\begin{equation}\label{haop}
    \Gamma x:=P_- (F_- x), \quad x\in H^2,
\end{equation}
has norm less or equal to one, $\|\Gamma\|\le 1$ (the operator is
 naturally defined,
say, on polynomials and then extended by continuity) .

Let
\begin{equation}\label{nehclass}
    \cN(F_-)=\{f\in L^\infty: F_-=P_-f,\ \|f\|_\infty\le 1\}.
\end{equation}
The Nehari {\it problem} deals with a description of $\cN(F_-)$
 for the given $F_-$. Thus the Nehari Theorem is the solvability
 condition for this problem.
 The problem was solved by Adamyan, Arov, and Krein \cite{AAK1, AAK2,
 AAK3}.
In the case of non uniqueness the set of solutions is parameterized
by the unite ball of the class $H^\infty$.  Precisely, there exists
 $\phi=\phi_{F_-}\in H^\infty$
with the following three properties
\begin{equation}\label{propphi}
    \|\phi\|_\infty\le 1, \quad \int_{\bbT}\log(1-|\phi|)\, dm>-\infty,
 \quad \phi(0)=0.
\end{equation}
This function is accompanied by  the outer function $\psi$
\begin{equation}\label{psi}
    \psi(\zeta)=e^{\frac 1 2\int_{\bbT}
 \frac{t+\zeta}{t-\zeta}\log(1-|\phi(t)|^2)\, dm(t)},
\end{equation}
and the function
\begin{equation}\label{f0}
    f_0=-\frac{\bar\phi\psi}{\bar\psi}.
\end{equation}
The set $\cN(F_-)$ is of the form
\begin{equation}\label{dcrp-nehclass}
    \cN(F_-)=\{f=f_{\cE}=f_0+\frac{\psi^2\cE}{1-\phi\cE}: \cE\in
 H^\infty,\ \|\cE\|_\infty\le 1\}.
\end{equation}

\vspace{.25in}

This is the hard
\begin{problem} \label{pr1} Specify analytic properties
of those holomorphic $\phi$'s of \eqref{propphi} that generate the
 description \eqref{dcrp-nehclass}.
Following to Arov \cite{Ar} we call such $\phi$'s regular.
 \end{problem}

\begin{remark}
It is convenient to associate with a function $\phi$ of the form
 \eqref{propphi}
the unitary valued matrix function
\begin{equation}\label{mf}
    \cS=\cS_\phi=\begin{bmatrix}
    \phi&\psi\\ \psi& f_0
    \end{bmatrix}
\end{equation}
with the entries given by \eqref{psi}, \eqref{f0}. Then the relation \eqref{dcrp-nehclass} between $f$ and $\cE$ can be rewritten into
the vector form
$$
\begin{bmatrix}
A\\f
\end{bmatrix}=\cS\begin{bmatrix}
A\cE\\1
\end{bmatrix}=\begin{bmatrix}
    \phi&\psi\\ \psi& f_0
    \end{bmatrix}\begin{bmatrix}
A\cE\\1
\end{bmatrix},
$$
where $A$ is defined by this relation in a unique way,
$A=\frac{\psi}{1-\phi\cE}$. The fact that $\cS$ is unitary implies that
$$
|f|^2+|A|^2=1+|A|^2|\cE|^2,
$$
i.e.:
$$
1-|f|^2=|A|^2(1-|\cE|^2)\ge 0.
$$

\end{remark}

\bigskip

Let us give an example of non-regular $\phi$ from \eqref{propphi}.
Choose any inner function $\Delta, \Delta(0) >0$.
The point is, that
 a {\it  holomorphic} matrix function
\begin{equation}\label{mf1}
    \begin{bmatrix}
    \frac{\Delta-\Delta(0)}{1+\Delta(0)}&
 \sqrt{\Delta(0)}\frac{1+\Delta}{1+\Delta(0)}
     \\ \sqrt{\Delta(0)}\frac{1+\Delta}{1+\Delta(0)}&
  \frac{\Delta-\Delta(0)}{1+\Delta(0)}
    \end{bmatrix},
\end{equation}
is unitary on $\mathbb{T}$. That is, for such
 $\phi= \frac{\Delta-\Delta(0)}{1+\Delta(0)}$
the corresponding $f_0=\phi$ belongs to $H^\infty$, and thus
the class
\begin{equation}\label{dcrp-nehclassce}
    \{f=f_0+\frac{\psi^2\cE}{1-\phi\cE}: \cE\in H^\infty,\
 \|\cE\|_\infty\le 1\}
\end{equation}
describes a {\em proper} subclass of the set $\cN(0)$=unit ball of $H^{\infty}$ (since $P_-f_0=0$ in
 this case).

 Let us show that, for instance $f=0$ can not be represented in this way. First let us note that
 $\frac{\psi}{1-\phi\cE}\in H^2$. In fact, since $\frac{1+\phi\cE}{1-\phi\cE}$ is a function
 in the unite disc with the positive real part we have (in the sense of the boundary values on the unite circle)
 \begin{equation}
 \label{mnogoraz}
     \left|\frac{\psi}{1-\phi\cE} \right|^2\le\frac{1-|\phi\cE|^2}{|1-\phi\cE|^2} =\Re \frac{1+\phi\cE}{1-\phi\cE}\in L^1,
 \end{equation}
and in addition  $\frac{\psi}{1-\phi\cE}$ is a function of the Smirnov class (the denominator is an outer function). So
if $0$ is in the set, we get
$$
\frac{f_0}\psi=-\frac{\psi\cE}{1-\phi\cE}\in H^2.
$$
 On the other hand this function belongs to $H^2_-$ due to the representation
$$
\frac{f_0}\psi=-\frac{\bar\phi}{\bar \psi}
$$
($\phi/\psi$ is also of the Smirnov class). Thus $\phi=0$ and $\psi=1$.

 \bigskip

Note that actually this is a general obstacle: according to the Arov's
 Theorem
one can always "factor out" in a certain sense a holomorphic
 $\cS$-matrix from the given one, so that the remaining part,
indeed, generate the description of a Nehari problem in the form
 \eqref{dcrp-nehclass}
(the, so called, singular-regular factorization \cite{Ar}).

On a ceratin stage the answers
to Problem \ref{pr1} and to  a comparably long list of similar problems
 (see for instance
\cite{BS1, BS2} where the
similar question with respect to the Hamburger moment problem is
 discussed) were formulated in terms
of density of a certain class of holomorphic  function in an associated
 with the data Hilbert space.

 We need to recall the Nagy--Foias functional model space \cite{NF}. It can be associated with an arbitrary
 function $\phi$ of the unite ball of $H^\infty$ (the Schur class)
 $$
 K_\phi:=H^2\oplus\overline{\Delta L^2}\ominus \{\phi\oplus\Delta\} H^2,
 $$
  where $\Delta:=\sqrt{1-|\phi|^2}$, and
  $$
  \overline{\Delta L^2}=\text{clos}_{L^2}\{f=\Delta g: g\in L^2\}.
  $$

In our specific case $\log(1-|\phi|^2)\in L^1$ we can chose an "analytic" square root instead of "arithmetic",
i.e., to use $\psi$ instead of $\Delta$, and of course $\overline{\Delta L^2}=L^2$. So, the functional space is of the form
$$
K_\phi:=\begin{bmatrix}H^2\\L^2\end{bmatrix}\ominus\begin{bmatrix}
    \phi \\ \psi
    \end{bmatrix} H^2=\begin{bmatrix}0\\H^2_-\end{bmatrix}\oplus \hat H,
$$
where
\begin{equation}\label{hhat}
    \hat H:=\hat H_\phi= H^2(\bbC^2)\ominus\begin{bmatrix}
    \phi \\ \psi
    \end{bmatrix} H^2.
\end{equation}
That is, we have $\begin{bmatrix}
    x_+ \\ g
    \end{bmatrix}\in K_\phi$ if and only if
    $
    x_+\in H^2$, $g\in L^2$
     and
     $$
     x_-:=\bar\phi x_++\bar\psi g\in H^2_-.
     $$

Alternatively, we can characterize $K_\phi$ as pairs $ \begin{bmatrix}
    x_+ \\ x_-
    \end{bmatrix}$ such that
$$
    x_\pm \in H^2_\pm \quad\text{and}\quad  g:=\frac{x_--\bar\phi x_+}{\bar\psi} \in L^2.
     $$
     It looks  natural to hope that the pairs
     \begin{equation}\label{dancey}
        \begin{bmatrix}
    x_+ \\ x_-
    \end{bmatrix} =\begin{bmatrix}
    \psi y_+ \\ \bar \psi y_-
    \end{bmatrix},\quad y_\pm\in H^2_\pm
     \end{equation}
     form a dense set in $K_\phi$ (recall $\psi$ is an outer function). The corresponding $g$,
     $$
     g=y_-- \frac{\bar\phi\psi}{\bar\psi} y_+=y_-+f_0y_+,
     $$
      for sure belongs to $L^2$ and the element of $K_\phi$ is of the form
       \begin{equation}\label{elde}
       \begin{bmatrix}\psi y_+\\ g\end{bmatrix} = \begin{bmatrix}
    0\\ y_-+P_-f_0 y_+
    \end{bmatrix} \oplus\begin{bmatrix}
    \psi y_+\\ P_+ f_0 y_+
    \end{bmatrix} \in \begin{bmatrix}0\\H^2_-\end{bmatrix}\oplus \hat H.
       \end{equation}

    However, in fact,
\begin{theorem}
\label{1.2}
A function $\phi$ is regular if and
 only if the vectors of the form \eqref{elde} form a dense set in $K_\phi$, or what is the same,
 \begin{equation}\label{dcond1}
\check H_\phi:={\rm clos}_{H^2(\bbC^2)}\left\{\begin{bmatrix}
    \psi y_+\\ P_+ f_0 y_+
    \end{bmatrix}:y_+\in H^2\right\} =\hat H_\phi.
 \end{equation}
Moreover,  \eqref{dcond1} holds as soon as
\begin{equation}\label{dcond2}
   P_+\bar t\begin{bmatrix}
    \phi(t) \\ \psi(t)
    \end{bmatrix} =\begin{bmatrix}
    \frac{\phi(t)} t \\ \frac{\psi(t)-\psi(0)}t
    \end{bmatrix}\in \check H_\phi.
\end{equation}
\end{theorem}
 For a proof  see e.g. \cite{Kh2}.

 \medskip

A trivial consequence is the following
\begin{proposition}
\label{trivial}
 Let
\begin{equation}\label{cr1}
    \frac 1{1-|\phi|^2}\in L^1.
\end{equation}
Then $\phi$ is regular.
\end{proposition}
Indeed, we put $x(t)=\frac{\phi(t)}{t\psi(t)}\in H^2$ and we get
\begin{equation*}
 \begin{bmatrix}
    \psi x\\ P_+ f_0 x
    \end{bmatrix}= \begin{bmatrix}
    \frac{\phi(t)}{t}\\ P_+  \frac{1}{t}(\psi-1/\bar\psi)
    \end{bmatrix}=\begin{bmatrix}
    \frac{\phi(t)} t \\ \frac{\psi(t)-\psi(0)}t
    \end{bmatrix},
\end{equation*}
since $\frac 1 t 1/\bar\psi\in H^2_-$.

One of the main goal of this note is to discuss: is it possible to give
 a better then \eqref{cr1}
sufficient condition in terms of the absolute value of $\phi$?

We have to point out on a nice result that was obtain in \cite{Kh1},
 see also \cite{Kats, Kh2}.
It was shown that there is no necessary and sufficient
condition of regularity of $\phi$ in terms of the absolute value
 $|\phi|$.
\begin{theorem}
Let $\phi\in H^\infty$ satisfies \eqref{propphi}. Then there exists an
 inner function $\Phi$ such that
$\phi\Phi$ is regular.
\end{theorem}

\section{Condition on modulus $|\psi|$ which ensures regularity of $\phi$ but which is weaker than $1/\psi \in H^2$.}
\label{chimney}

We want to see some non-trivial conditions on $|\phi|$ that guarantee that $\phi$ is regular.
By non-trivial we understand any condition different from
\begin{equation}
\label{triv}
\frac1{|\psi|^2}=\frac1{1-|\phi|^2} \in L^1\,.
\end{equation}
We denote by $h$ the outer function with modulus
$$
|h|^2:= \frac1{1-|\phi|^2} \,.
$$
In other words
$$
h=\frac1{\psi}\,.
$$
Outer $h$  always exists by the assumption \eqref{propphi} on $\phi$. It is {\em not} in $H^2$ throughout this section because we are looking for
``non-triviality".

\begin{theorem}
\label{log4}
Suppose  $|h|^2:= \frac1{1-|\phi|^2} \notin L^1(\mathbb{T})$. Suppose also that
\begin{equation}
\label{Nn1}
\liminf_{N\rightarrow\infty}(\int_{|h|\le N} |h|^4 dm)(\int_{|h|> N}(\log |h|)^4 dm)= 0\,,
\end{equation}
Then $\phi$ is regular.
\end{theorem}

\begin{corollary}
\label{log4cor}
Suppose  $|h|^2:= \frac1{1-|\phi|^2} \notin L^1(\mathbb{T})$. Suppose also that
\begin{equation}
\label{Nn2}
\liminf_{N\rightarrow\infty}N^4\,\int_{|h|> N}(\log |h|)^4 dm)= 0\,,
\end{equation}
Then $\phi$ is regular.
\end{corollary}

\vspace{.2in}

Let us explain a bit assumptions  \eqref{Nn1}, \eqref{Nn2}.  It is easy to to fulfill them if there exists a sequence $N_n\rightarrow\infty$ such that

\begin{equation}
\label{levelsets}
|\{\zeta\in \mathbb{T}: |h| > N_n\}|( N_n\log N_{n+1})^4 \le \frac1{n^2}\,.
\end{equation}

On the other hand, this is easily reconcilable with the following condition which guarantees non-triviality:

\begin{equation}
\label{nontriv}
|\{\zeta\in \mathbb{T}:  |h|=N_{n+1}\}|( N_{n+1})^2 \ge (n+1)^2\,.
\end{equation}

In fact, to have both \eqref{nontriv} and \eqref{levelsets} one can define $|h|$ to be step-function having values $N_n, n=1, 2, 3,...$
on sets having measures $\frac{n^2}{N_n^2}$ (this gives \eqref{nontriv}) , and choose
$N_n$ going to infinity extremely fast to have  firstly $ \sum_{k=n+1}^{\infty}\frac{k^2}{N_k^2}< \frac{2(n+1)^2}{N_{n+1}^2}$ and secondly
$$
\frac{(n+1)^2}{N_{n+1}^2} ( N_n\log N_{n+1})^4 \le \frac1{2n^2}\,.
$$
Then \eqref{levelsets} follows.

\vspace{.2in}

\noindent{\bf Remark.}
 Conditions \eqref{Nn1}, \eqref{Nn2} are of course the condition just on $|\phi|$, or, which is the same, on $|\psi|$.

  \vspace{.2in}

 Now we will prove Theorem \ref{log4}.  We are grateful to A. Aleksandrov whose idea is used in the proof.

 \bigskip

 \begin{proof}

We need to use \eqref{Nn1} to prove
the existence of $H^2$ functions $v_n$ such that
$$
\psi v_n  \rightarrow \frac{\phi(\zeta)}{\zeta}\,\,\,\text{in}\,\,\, H^2\,\,\,\text{and}
$$
$$
-P_+\bigg(\frac{\bar{\phi}\psi}{\bar{\psi}}v_n\bigg) \rightarrow \frac{\psi(\zeta)- \psi(0)}{\zeta}\,\,\,\text{in}\,\,\, H^2\,.
$$
Notice that if we would  have $\frac1{|\psi|^2} = |h|^2 \in L^1$ then we could have taken
$$
v_n = \frac{\phi(\zeta)}{\zeta \psi(\zeta)}
$$
which would have been functions in $H^2$ in this case. But we have exactly opposite case:  $\frac1{|\psi|^2} = |h|^2 \notin L^1$.

Notice that to satisfy the above relationships it is enough to build $g_n\in H^2$ such that

\begin{equation}
\label{gnh}
g_n h \in H^2\,\,\,\forall n\,,
\end{equation}
and such that
\begin{equation}
\label{gn}
g_n \rightarrow 1\,\,\,\text{in}\,\,\, H^2\,\,\,
\text{and}\,\,\,-P_+\bigg(\frac{|\phi|^2\bar{\zeta}}{\bar{\psi}}g_n\bigg) \rightarrow \frac{\psi(\zeta)- \psi(0)}{\zeta}\,\,\,\text{in}\,\,\, H^2\,.
\end{equation}

In fact, having $g_n$ like that we put $v_n= \frac{\phi(\zeta)}{\zeta}  h  g_n$. Then $v_n\in H^2$ by \eqref{gnh}.
And these $v_n$ satisfy two conditions mentioned above because of \eqref{gn}.

Now let us write
$$
-P_+\bigg(\frac{|\phi|^2\bar{\zeta}}{\bar{\psi}}g_n\bigg)= -P_+(\bar{\zeta}\frac{g_n}{\bar{\psi}}) + P_+(\bar{\zeta} \psi g_n) =: I_n +II_n\,.
$$
 Now
 $$
 II_n = \frac{\psi g_n - (\psi g_n)(0)}{\zeta} \rightarrow \frac{\psi-\psi (0)}{\zeta}
 $$
 in $H^2$ because $g_n\rightarrow 1$ in $H^2$, $\psi$ is from $H^{\infty}$ and backward shift operator is bounded in $H^2$.

 So the only thing we need now is to construct $g_n$ such that \eqref{gnh} holds, $g_n\rightarrow 1$ in $H^2$, and
 \begin{equation}
 \label{I}
 -I_n = P_+(\bar{\zeta}\bar{h}g_n) \rightarrow 0
 \end{equation}
 in $H^2$.

 To have all this it is enough to have
 \begin{equation}
 \label{I1}
 P_+(\bar{\zeta}\bar{h}g_n) \rightarrow 0\,\,\,\text{in}\,\,\,H^2\,;\,\,\, g_nh \in H^2\,;\,\,\, g_n\rightarrow 1\,\,\,\text{in}\,\,\,H^2\,.
 \end{equation}

Let us fix a sequence $N_n\rightarrow\infty$, put
$$
\phi_n =\begin{cases} \log |h|,\,\,\,|h|> N_n\\
0,\,\, |h|\le N_n\end{cases}
$$
Here are our
$$
g_n:= e^{-(\phi_n +i\widetilde{\phi_n)}}
$$

In fact, obviously,
\begin{equation}
\label{hphi}
\|P_+(\bar{\zeta}\bar{h}g_n)\|^2_2 \leq \|\bar{h} e^{-(\phi_n +i\widetilde{\phi_n)}} - \bar{h} e^{-(\phi_n -i\widetilde{\phi_n)}}\|^2_2\,.
\end{equation}

And, of course, $|g_n||h| \le N_n$.

Now we  have $\|\bar{h} e^{-(\phi_n +i\widetilde{\phi_n)}} - \bar{h} e^{-(\phi_n -i\widetilde{\phi_n)}}\|^2_2=4
 \|h e^{-\phi_n} \sin\widetilde{\phi_n}\|_2^2$. In conjunction
with \eqref{hphi} this gives (we also use that $|\sin x|\le |x|$ and the definition of $\phi_n$)

\begin{equation}
\label{hphi1}
\frac 1 4\|P_+(\bar{h}g_n)\|^2_2 \leq \|he^{-\phi_n} \widetilde{\phi_n}\|^2_2\le \int_{|h|\le N_n} |h|^2 |\widetilde{\phi_n}|^2 + \int_{\mathbb{T}} |\widetilde{\phi_n}|^2 =: J_1 +J_2.
\end{equation}

 To estimate $J_1$ we write (using the boundedness of the harmonic conjugation operator $\widetilde{\cdot}$ in $L^4(\mathbb{T})$)
 \begin{multline}
\label{hphi2}
J_1\le \int_{|h|\le N_n} |h|^2 |\widetilde{\phi_n}|^2 \le (\int_{|h|\le N_n} |h|^4 dm)^{\frac12} (\int_{\mathbb{T}} |\widetilde{\phi_n}|^4dm)^{\frac12} \le  \\ C\,(\int_{|h|\le N_n} |h|^4 dm)^{\frac12} (\int_{\mathbb{T}} |\phi_n|^4dm)^{\frac12}= \\ C\,(\int_{|h|\le N_n} |h|^4 dm)^{\frac12} (\int_{|h|>N_n} |\log |h||^4dm)^{\frac12}\,.
\end{multline}

In particular, if \eqref{Nn1} holds, there exists a sequence of numbers $N_n\rightarrow \infty$ such that the last expression tends to zero. Thus $J_1\rightarrow 0$.

Now let us estimate $J_2$.

 \begin{multline}
\label{hphi3}
J_2=\int_{\mathbb{T}} |\widetilde{\phi_n}|^2  \le (\int_{\mathbb{T}} |\widetilde{\phi_n}|^4\,dm)^{\frac12} \le C\,(\int_{\mathbb{T}} |\phi_n|^4\,dm)^{\frac12} =\\
C\,(\int_{|h|>N_n} |\log|h||^4\,dm)^{\frac12} \le C\,(\int_{|h|\le N_n} |h|^4 dm)^{\frac12} (\int_{|h|>N_n} |\log |h||^4dm)^{\frac12}\,.
\end{multline}
The last inequality holds  for large $N_n$. In fact, for large $N_n$ integral  $\int_{|h|\le N_n} |h|^4 dm$ is as large as we wish because we assumed that $\int_{\mathbb{T}} |h|^2 dm =\infty$.

Therefore, if \eqref{Nn1} holds then $J_2\rightarrow 0$. Going back to \eqref{hphi} we see that we proved Theorem \ref{log4}.

 \end{proof}

\section{Strong Regularity}

The strong regularity means that  $\phi$ is regular and in addition $\|\Gamma\|<1$, where
 $\Gamma=\Gamma_{f_0}$.  In other words, $\phi$ is strongly regular if and only if $\phi$ is regular and
the operator
$(I-\Gamma^*\Gamma)$ is invertible.

The following theorem is a combination the  Helson-Szeg\"o, Hunt-Muckenhoupt-Wheeden, and Adamyan-Arov-Krein Theorems (AAK), for a matrix generalization
see e.g. \cite{AD1, AD2}.

\begin{theorem}
\label{strongly}
Function $\phi$ is strongly regular if and only if
\begin{equation}
\label{sa2}
\frac{1-|\phi|^2}{|1-\phi |^2} = w \in A_2.
\end{equation}
\end{theorem}
Note that the left hand side in \eqref{sa2} being a positive harmonic function in the unit disc ($=\Re\frac{1+\phi}{1-\phi}$)
is equal to the harmonic extension of an $A_2$ weight on the circle.

\begin{proof} Let $\phi$ is regular and $\|\Gamma_{f_0}\|<1$.
Consider the symbol $f_1$, which corresponds to
the choice $\cE=1$ in \eqref{dcrp-nehclass} (recall $f_0$ corresponds to $\cE=0$).
It is of the form
$$
f_1=\frac{\bar{g}}{g}, \quad g:= \frac{1-\phi}{\psi}
$$
and we have $\Gamma=\Gamma_{f_1}$. By the Adamyan-Arov-Krein Theorem (AAK), see Remark \ref{raak},
\begin{equation}\label{aakpsi}
     \frac{1}{\psi\psi(0)}=(I-\Gamma^*\Gamma)^{-1}1\in H^2.
\end{equation}
Thus $g\in H^2$ and $\|\Gamma_{\frac{\bar g}g}\|<1$. By the  Helson-Szeg\"o and Hunt-Muckenhoupt-Wheeden Theorems, see e.g. \cite{Nik},
$$
|g|^2 = \frac{|1-\phi|^2}{1-|\phi|^2}\in A_2.
$$

Conversely, from \eqref{sa2} we conclude
\begin{equation}
\label{a21}
\frac{1+\phi}{1-\phi} = w +i \widetilde{w}\,,
\end{equation}
as before  $\widetilde{w}$ stands for the harmonic conjugate of $w$ (the Hilbert transform on the circle).

Then of course
$$
\frac{1-\phi}{1+\phi} =\frac{1}{ w +i \widetilde{w}}\,,
$$
and so
\begin{equation}
\label{a22}
\frac{1-|\phi|^2}{|1+\phi |^2} = \Re \frac{1}{ w +i \widetilde{w}} = \frac{w}{w^2 + \widetilde{w}^2}\,.
\end{equation}

Let us derive that
\begin{equation}
\label{a23}
\frac{|1+\phi|^2}{1-|\phi |^2} \in L^1
\end{equation}
on the unit circle.
Indeed,
$$
\int_{\bbT}\frac{|1+\phi|^2}{1-|\phi |^2}  = \int_{\bbT}w + \int_{\bbT}\frac{\widetilde{w}^2}{w}\le \|w\|_1 + Q_{1/w}\|w\|_1<\infty\,,
$$
where $Q_{1/w}$ stands for the norm of the Hilbert transform from $L^2_{1/w}$ to itself, which is finite as $w\in A_2$.

Combine \eqref{a23} with a simple remark that \eqref{sa2} implies  $\frac{|1-\phi|^2}{1-|\phi |^2} \in L^1(\bbT)$. Add these two relations and obtain
\begin{equation}
\label{aF}
\frac{1}{1-|\phi|^2} \in L^1(\bbT)\,.
\end{equation}
We know that this is sufficient for being regular.

Finally, by the converse statement in the
 Helson-Szeg\"o and Hunt-Muckenhoupt-Wheeden Theorems we have $\|\Gamma\|=\|\Gamma_{f_1}\|=\|\Gamma_{\frac{\bar{g}}{g}}\|<1$, if $|g|^2 \in A_2$. The latter is exactly \eqref{sa2}.
\end{proof}

\begin{remark}\label{raak}
Let us comment \eqref{aakpsi} from the point of view of regularity. We still assume that $\|\Gamma\|<1$,
that is $(I-\Gamma^*\Gamma)^{-1}1$ has the direct meaning.
 Then it is easy to check that the vector
\begin{equation}\label{kcheck}
    \check k:=\begin{bmatrix}\psi (I-\Gamma^*\Gamma)^{-1}1\\P_+f_0(I-\Gamma^*\Gamma)^{-1}1
    \end{bmatrix}\psi(0)
\end{equation}
is the reproducing kernel in $\check H_\phi^2$, see \eqref{dcond1}. Indeed,
$$
\langle\begin{bmatrix}\psi y_+\\P_+ f_0 y_+
    \end{bmatrix},\check k\rangle=\langle(I-\Gamma^*\Gamma)y_+,(I-\Gamma^*\Gamma)^{-1}1\rangle\psi(0)=
    y_+(0)\psi(0).
    $$
Also, it is evident that the reproducing kernel of $\hat H_\phi^2$ is $\hat k=\begin{bmatrix}1\\ 0\end{bmatrix}$, i.e.:
$$
\langle\begin{bmatrix} x_1\\ x_2\end{bmatrix},\hat k\rangle= x_1(0),\quad \forall \begin{bmatrix} x_1\\ x_2\end{bmatrix}\in \hat H_\phi^2.
$$
Thus, $\hat H_\phi=\check H_\phi$ implies $\hat k=\check k$, and the equality of the first components is precisely
\eqref{aakpsi}.

Generally, in AAK theory, for a regular $\phi$ the following formula holds
\begin{equation*}
    \frac{1}{\psi(\zeta)\psi(0)}=\lim_{r\downarrow 1}((rI-\Gamma^*\Gamma)^{-1}1)(\zeta),
\end{equation*}
so the last function is not necessary in $H^2$.

\end{remark}

\subsection{Less than one.}
\label{mensheedinicy}

The previous proof exploited a lot of AAK theory in its part that proves \ref{sa2} from strong regularity, and we wish to give a more direct proof for the reader who is not so familiar with this subtle material.

\noindent{\em The second proof.} 

First we need an AAK  lemma, which can be found by the reader in \cite{Nik} or extracted from AAK papers from our references list (however we provide the proof for the sake of completeness).

\begin{lemma}
\label{menshe1}
Let $F\in L^{\infty}$ and $d (F, H^{\infty}) <1$. Then the coset $F+H^{\infty}$ contains a function $\frac{\bar{h}}{h}$, where $|h|^2 \in A_2$.
In particular, this coset contains the unimodular function $v$ such that the Toeplitz oerator $T_v$ is invertible.
\end{lemma}

 \begin{proof}
Let $d(f, H^{\infty})$ denote the distance between a function $f\in L^{\infty}$ and the subspace $H^{\infty}$ of bounded holomorphic functions in the unit disc.  Let $H^{\infty}_0$ denote bounded holomorphic functions in the unit disc that vanish at zero. Given such an $F$ consider $d (\bar{z}F, H^{\infty}) = d (F, H^{\infty}_0)$.
 Two cases may happen.  Suppose first that $d (\bar{z}F, H^{\infty}) =1$.  Then operator Hankel operator $H_{\bar{z}F}:= P_- (\bar{z}F\cdot): H^2\rightarrow H^2_-$ attains its norm.
 In fact, $\|H_{\bar{z}F}\|_{ess}\le d (\bar{z}F, \bar{z}H^{\infty}) = d (F, H^{\infty} )<1= d (\bar{z}F, \bar{z}H^{\infty}) = \|H_{\bar{z}F}\|$. This is just classical Nehari's theorem (see \cite{Nik}), and $\|\cdot\|_{ess}$ means the norm modulo compact operators (essential norm). If the essential norm of the operator $A$ in the Hilbert space is strictly less than its norm, then $A$ attains its norm. See \cite{Nik} Ch VII again, or just notice that
 we can reduce our statement to self-adjoint operators by considering $A^*A$ (and polar decomposition $A=U(A^*A)^{1/2}$).  But if the essential norm of the self adjoint operator is strictly smaler than its norm, it means that its norm is a maximal eigenvalue od finite multiplicity (spectral theorem), and, thus, the operator attains its norm. Now let our Hankel operator attain its norm $1$ at vector $H\in H^2, \|H\|_2=1$. Denote $\bar{G}_0 = H_{\bar{z}F}H$. Denote by $u$ a function in the coset $\bar{z}F, H^{\infty}$ of $\|u\|_{\infty}=1$. It always exists by obvious compactness argument.
 
 Then
 $$
 1= \|\bar{G}_0\|_2 = \|H_{\bar{z}F}H\|= \|H_u H\|_2 = \|P_ (u H)\| \le \|uH\|_2 \le \|H\|_2 =1\,.
 $$
 This means of course that $|u|=1$ almost everywhere on the circle and that $uH$ is antianalytic (that is $P_+ (uH)=0$).
 
 Therefore,
 
 $$
 uH = P_- (uH) = \bar{G}_0\,.
 $$
 We conclude that  two $H^2$-functions $H$ and $G_0$ have the same modulus a. e. on the unit circle. Write $H=S_1 h $, $G_0 = z S_2 h$ their inner-outer factorizations ($h$ is an outer function here).
 
 Then we obtain 
 $$
 \frac{\bar{h}}{S_0 h}=u \in \bar{z}F + H^{\infty}\,,
 $$
 where $S_0 =z S_1 S_2$, $h\in H^2$.
 Consider $S =z S_1 S_2$ and
 $$
 v:= z u =  \frac{\bar{h}}{S h} \in F + H^{\infty}\,.
 $$
  
  We want to prove now that Toeplitz operator $T_v$ is invertible.
 From the fact that $d (v,H^{\infty})= d (F, H^{\infty}<1$ and from Nehari's theorem we know that $\|H_v\|<1$. But   $|v|=1$ a. e. on the unit circle and then $T_v^*T_v = I - H_v^*H_v$. Therefore $\|H_v\|<1$ means that $T_v$ is bounded from below (that is it is left-invertible). To prove that it is invertible it is sufficient to prove that its adjoint has only trivial kernel. Let $R\in Ker T^*_v = T_{\bar{v}}$.
 Then
 $$
 T_{\bar{v}}R =0\Rightarrow P_+ ( \frac{h}{\bar{S}\bar{h}}R) =0 \Rightarrow   \frac{h}{\bar{S}\bar{h}}R = \bar{z}\bar{r}\,,
 $$
 where $r\in H^2$. Then $ hR = \bar{z} \bar{S}\bar{h}\bar{r}$, the left hand side being from $H^1$, and the right hand side being from $H^1_-$. The intersection being zero we conclude that $R=0$.
 
 So we get $v$, $v:= z u =  \frac{\bar{h}}{S h} \in F + H^{\infty}$ such that $T_v$ is invertible. By Helson-Szeg\"o theorem (see \cite{Nik}, Ch. VII) we conclude that $v=\frac{\bar{h}}{h},$ where $|h|^2 \in A_2$.
 
\vspace{.1in}

Now we need to consider the second case: $d (\bar{z}F, H^{\infty}) <1$.  We denote $f:= \bar{z}F$ and consider the function $\tau(c):= d (f + c\bar{z}, H^{\infty})$.  We know that $\tau(0) <1, \tau(\infty) =\infty$ and $\tau$ is obviously continuous. So we can find $c_0$ such  that for $ f+ c_0\bar{z}=:\bar{z}\Phi$
$$
d (\bar{z}\Phi, \bar{z}H^{\infty}) = d (f+ c_0\bar{z}, \bar{z}H^{\infty})= d (f, H^{\infty}) <1\,,
$$
$$
d (\bar{z}\Phi, H^{\infty}) = d (f+ c_0\bar{z}, H^{\infty})=1\,.
$$
 The we proceed exactly as in the first case by using the fact that the last two relationships imply that operator $H_{\bar{z}\Phi}$ attains its norm. We will find unimodular $u\in f+c_0\bar{z}+ H^{\infty}$ such that
 $$
 u=\frac{\bar{g}}{z\Theta g}\,,
 $$
 where $\Theta $ is inner and $g$ is outer from $H^2$. Therefore, $v:= z u$ will be in coset $zf + H^{\infty} = F+ H^{\infty}$ and will have the form $u= =\frac{\bar{g}}{\Theta g}$ wit the same $\Theta$ and $g$.
 
 Again as in the first case $\|H_v\|= d (F, H^{\infty}) <1$ (Nehari's teorem) ensures that $T_v$ is left-invertible. And exactly as before we prove that $T_v^*=T_{\bar{v}}$ has a trivial kernel.
 Hence $T_v$ is invertible and we conclude once again by Helson-Szeg\"o theorem (see \cite{Nik}, Ch. VII)  that $v=\frac{\bar{h}}{h},$ where $|h|^2 \in A_2$.
AAK lemma is proved.
\end{proof}

It is easy to finish the second proof of our theorem. Let $\phi$ be strongly regular. This means that it is reguar and so $f_0:= -\frac{\bar{\phi}\psi}{\psi}$ is such that
$$
f_0 + H^{\infty} = f_0 +\frac{\psi^2 e}{1-\phi e}\,,
$$ 
where $e$ runs over the unit ball of $H^{\infty}$.  But strong regularity means also that $\|H_{f_0}\|<1$. Lemma \ref{menshe1} means that there exists an outer $h$ such that $\frac{\bar{h}}{h} \in f_0 + H^{\infty}$ and $|h|^2 \in A_2$.
We gather:
$$
\frac{\bar{h}}{h} =-\frac{\bar{\phi}\psi}{\psi} +\frac{\psi^2 e}{1-\phi e}
$$ 
for some $e$ from  the unit ball of $H^{\infty}$.  Then a. e. on the circle
$$
\bigg| -\frac{\bar{\phi}\psi}{\psi} + \frac{\psi^2 e}{1-\phi e}\bigg| =1\,.
$$
But the left hand side is
$$
\bigg| \frac{e-\phi}{1-\phi e}\bigg|\,,
$$
and we conclude that$e$ is an inner function. Then
$$
\frac{\bar{h}}{h} =-\frac{\bar{\phi}\psi}{\psi} +\frac{\psi^2 e}{1-\phi e} = \frac{\psi \overline{1-\phi e}}{\bar{\psi}(1-\phi e)}\,.
$$
Denote $g_e:= \frac{1-\phi e}{\psi}$. Then we just got 
$$
\frac{\bar{h}}{h}= e\frac{\bar{g_e}}{g_e}\,.
$$
But $|1/g_e|^2= \frac{1-|\phi e|^2}{|1-\phi e|^2}\in L^1$, see \eqref{mnogoraz},  and so $1/g_e \in H^2$ ($g_e$ is obviously an outer function). We then immediately conclude from the last display equality that $g=c\cdot h$ with some constant $c$. In fact, we have 
$$
e\frac{h}{g_e} =\frac{\bar{h}}{\bar{g_e}}\,,
$$ where the left hand side belongs to $H^1$ and the right hand side belongs to $H^1_-$. So both expressions are just constants. And $e=e^{ir}$, $g_e=c\cdot h$, as it has been promised.
Therefore, $ \frac{|1-e^{ir}\phi|^2}{1-|\phi|^2} \in A_2$.

Now we use Theorem \ref{curious} proved in the Appendix to conclude that $ \frac{|1-\phi|^2}{1-|\phi|^2} \in A_2$. 

We finished the proof that regularity implies \ref{sa2}.

\bigskip
The following criteria was proposed in \cite{AD2}.

\begin{theorem}\label{thAD} $\phi$ is strongly regular if and only if
the matrix weight
\begin{equation}\label{weight}
    W:=\begin{bmatrix}1&\phi\\ \bar\phi&1\end{bmatrix}
\end{equation}
satisfies matrix $A_2$ condition.
\end{theorem}

Matrix $A_2$ condition was found in \cite{TV}. Let us notice that Theorem  \ref{strongly} has the following curious corollary.

\begin{theorem}
\label{curious} 
Let $w$ be such that \eqref{sa2} holds.
 Consider a new positive harmonic function given by
\begin{equation}
\label{A2c}
w_{e^{ic}}:=\frac{1-|\phi|^2}{|1-e^{ic}\phi|^2}.
\end{equation}
The positive harmonic function $w_{e^{ic}}$  does not have
singular part in its Herglotz representation. It absolutely continuous part on the circle (also called $w_{e^{ic}}$) is uniformly in $A_2$.
\end{theorem}

It is absolutely trivial from the point of view of Theorem \ref{thAD}: 
If matrix function $W$ is in matrix $A_2$ then the following matrix is obviously also in matrix $A_2$:
$$
W_c:=\begin{bmatrix}e^{ic/2}&0\\ 0&e^{-ic/2}\end{bmatrix}\begin{bmatrix}1&\phi\\ \bar\phi&1\end{bmatrix}\begin{bmatrix}e^{-ic/2}&0\\ 0&e^{ic/2}\end{bmatrix}
$$ 
is just $W_c = \begin{bmatrix}1&\phi_c\\ \bar\phi_c&1\end{bmatrix}$, where $\phi_c = e^{ic}\phi$.


The matrix $A_2$ condition for the specific weight \eqref{weight} can
 be given as the following scalar condition
(this is a slight modification of the condition given in \cite{AD2}).
\begin{lemma} A weight of the form \eqref{weight} satisfies $A_2$ if
 and only if
\begin{equation}\label{sina2}
\sup_{I}\frac{1}{|I|}\int_{I} \frac{|\phi-\Mth|^2+
(1-|\Mth|^2)}{1-|\phi|^2} \,dm<\infty,
\end{equation}
where for an arc $I\subset \bbT$ we put
\begin{equation}
\Mth:=\frac{1}{|I|}\int_{I} \phi\,dm.
\end{equation}
\end{lemma}

\begin{proof} The matrix
weight $\begin{bmatrix}1&\phi\\
\bar\phi&1\end{bmatrix}$ is in $A_2$ implies
\begin{equation}\label{aa3}
\begin{split}
   \frac 1 {|I|} \int_I \frac{\begin{bmatrix}1&-\phi\\
-\bar\phi&1\end{bmatrix}}{1-|\phi|^2}dm\le &C\begin{bmatrix}1&\Mth\\
\bar{\Mth}&1\end{bmatrix}^{-1}\\
=&C\left\{\begin{bmatrix}1&0\\
\bar{\Mth}&\sqrt{1-|\Mth|^2}\end{bmatrix}\begin{bmatrix}1&\Mth\\
0&\sqrt{1-|\Mth|^2}\end{bmatrix}\right\}^{-1},
\end{split}
\end{equation}
or
\begin{equation}\label{aa4}
   \frac 1 {|I|}\int_I \frac{\begin{bmatrix}1-|\phi|^2+|\phi-\Mth|^2&(\Mth-\phi)\sqrt{1-|\Mth|^2}\\
\overline{(\Mth-\phi)}\sqrt{1-|\Mth|^2}&
1-|\Mth|^2\end{bmatrix}}{1-|\phi|^2}dm\le C,
\end{equation}
which is equivalent to \eqref{sina2}.
\end{proof}
\begin{remark}  The latter form \eqref{sina2} of condition  \eqref{sa2}  makes indeed evident that this condition
  is invariant with respect
to the rotation $\phi\mapsto \phi e^{i c}$.
In fact, the condition is stable
 with respect to an arbitrary
fraction--linear transform
$$
\phi(\zeta)\mapsto e^{ic}\frac{\phi(\zeta)-a\zeta}{1-\bar a \phi(\zeta)/\zeta}, \quad |a|<1, \ c\in\bbR.
$$
\end{remark}

Since a proof of Theorem \ref{thAD} is fairly simple and short we give it here.

\begin{proof} First we note that  \eqref{sina2} in case $I=\bbT$ has the form
\begin{equation}\label{sina3}
\int_{\bbT} \frac{|\phi|^2+1}{1-|\phi|^2} \,dm<\infty,
\end{equation}
therefore by Proposition \ref{trivial} such $\phi$ is regular.

Now, recall that a weight is in $A_2$ if and only if (there
 exists $Q>0$)
\begin{equation}\label{condform}
    \langle W^{-1}P_+X,P_+X\rangle\le Q\langle W^{-1} X, X\rangle, \
 \forall X\in L^2_{W^{-1}}.
\end{equation}
Note that
$$
H^2\to W\begin{bmatrix} 0\\H^2\end{bmatrix}\quad H_-^2\to
 W\begin{bmatrix} H_-^2\\0\end{bmatrix}
$$
are unitary embedding in $L^2_{W^{-1}}$.

The orthogonal complement to
 their sum (an alternative definition of $K_\phi$) consists of the vectors of the form
$$
K_\phi=\{X=\begin{bmatrix} x_+\\x_-\end{bmatrix}, \ x_\pm\in H^2_\pm\}.
$$
The vectors
$$
X=\begin{bmatrix} \psi x_+\\ \bar\psi x_-\end{bmatrix}
$$
evidently belongs to $K_\phi$. They are {\em dense} in $K_\phi$ if and only if
$\hat H_\phi=\check H_\phi$, i.e. $\phi$ is regular, Theorem \ref{1.2}.

Now we calculate the quadratic forms \eqref{condform} for the test--vectors
 $$
X=\begin{bmatrix} \psi x_+\\ \bar\psi x_-\end{bmatrix}+W\begin{bmatrix}
 y_-\\y_+\end{bmatrix},
\quad x_\pm\in H^2_\pm,\  y_\pm\in H^2_\pm,
$$
We get
\begin{equation*}\label{form1}
\begin{split}
    \langle W^{-1} X, X\rangle=&\langle \frac{\begin{bmatrix}1&-\phi
    \\  -\bar\phi&1\end{bmatrix}}{1-|\phi|^2} \begin{bmatrix} \psi
 x_+\\ \bar\psi x_-\end{bmatrix},
    \begin{bmatrix} \psi x_+\\ \bar\psi x_-\end{bmatrix}\rangle+
    \langle \begin{bmatrix}1&\phi\\ \bar\phi&1\end{bmatrix}
 \begin{bmatrix} y_-\\y_+\end{bmatrix}, \begin{bmatrix}
 y_-\\y_+\end{bmatrix}\rangle\\
    =&\langle {\begin{bmatrix}1&\bar f_0
    \\  f_0&1\end{bmatrix}} \begin{bmatrix} x_+\\  x_-\end{bmatrix},
    \begin{bmatrix}  x_+\\ x_-\end{bmatrix}\rangle+\langle
 y_-,y_-\rangle
    +\langle y_+,y_+\rangle\\
    =&\langle {\begin{bmatrix}1& \Gamma^*
    \\  \Gamma&1\end{bmatrix}} \begin{bmatrix} x_+\\  x_-\end{bmatrix},
    \begin{bmatrix} x_+\\  x_-\end{bmatrix}\rangle+\langle
 y_-,y_-\rangle
    +\langle y_+,y_+\rangle.
    \end{split}
\end{equation*}
Since
$$
P_+X=\begin{bmatrix} \psi x_+\\ 0\end{bmatrix}+\begin{bmatrix} \phi\\
 1\end{bmatrix}y_+,
$$
we get
\begin{equation*}\label{form2}
\begin{split}
\langle W^{-1}P_+X,P_+X\rangle &=\langle \begin{bmatrix} 1\\-\bar
 \phi\end{bmatrix}\frac{x_+}{\bar\psi}+
\begin{bmatrix} 0\\ 1\end{bmatrix}y_+,\begin{bmatrix} \psi x_+\\
 0\end{bmatrix}+\begin{bmatrix} \phi\\ 1\end{bmatrix}y_+\rangle
\\
&=\langle x_+,x_+\rangle+\langle y_+,y_+\rangle.
  \end{split}
\end{equation*}
Therefore for such vectors \eqref{condform} is equivalent to
$$
\|x_+\|^2+\|y_+\|^2\le Q\{\langle(I-\Gamma^* \Gamma)x_+,x_+\rangle+
\|\Gamma x_+ +x_-\|^2+\|y_-\|^2+\|y_+\|^2\},
$$
which is evidently equivalent to $(I-\Gamma^* \Gamma)>\epsilon I$ with
 $\epsilon>0$.

 Thus,  $A_2$ implies the regularity and the bound $\|\Gamma\|<1$. Conversely,
 if $\phi$ is regular and  $\|\Gamma\|<1$ then \eqref{condform} holds on a dense set, and hence
 $W\in A_2$.
\end{proof}

It is interesting  that a similar condition on an arbitrary symbol of the Hankel
 operator guaranties invertibility of
$(I-\Gamma^*\Gamma)$.

\begin{proposition} Let $f$ satisfies the $A_2$-kind condition
\begin{equation}\label{sina21}
\sup_{I}\frac{1}{|I|}\int_{I} \frac{|f-\Ms|^2+
(1-|\Ms|^2)}{1-|f|^2} \,dm<\infty,
\end{equation}
where
\begin{equation}
\Ms:=\frac{1}{|I|}\int_{I} f\,dm.
\end{equation}
Then $\|\Gamma\|<1$ for $\Gamma x:= P_-fx$.
\end{proposition}

\begin{proof}
We have
\begin{equation*}
    \langle\begin{bmatrix} 1& \bar f\\ f&
 1\end{bmatrix}P_+X,P_+X\rangle\le Q
     \langle\begin{bmatrix} 1& \bar f\\ f& 1\end{bmatrix}X,X\rangle.
\end{equation*}
Put here
$$
X=\begin{bmatrix} x\\-\Gamma x\end{bmatrix}, \ x\in H^2.
$$
We get
$$
\langle x, x\rangle\le Q\langle (I-\Gamma^*\Gamma)x,x\rangle.
$$
\end{proof}

\vspace{.2in}

\noindent{\bf Appendix.}

\vspace{.1in}

We wish to prove Theorem \ref{curious} without relating to strongly regular functions in the sense of Arov-Dym, namely, to give a direct proof  of this curious result.

Let us recall that we start with $w\in A_2$. We extend it into the disc by harmonicity, get a positive harmonic function and represent it--as usual--in the form
\begin{equation}
\label{A2}
\frac{1-|\phi|^2}{|1-\phi|^2} = w\,,
\end{equation}
where $\phi$ is holomorphic in the disc and of $H^{\infty}$-norm at most $1$. Such functions $\phi$ and positive harmonic functions are in one to one correspondance  by \eqref{A2}.

Let $\phi_c:= e^{ic} \phi\,, \psi_c:= e^{ic} \psi\,,c\in\bbR$. Consider a new positive harmonic function given by
\begin{equation}
\label{A2c}
w_c:=\frac{1-|\phi_c|^2}{|1-\phi_c|^2}\,,
\end{equation}

\begin{theorem}
\label{curios}
If  $w \in A_2$, and $\phi\in H^{\infty}$, $\|\phi\|_{\infty}\le 1$ is given by \eqref{A2},  then a positive harmonic function $w_c$ from \eqref{A2c} does not have singular part in its Herglotz representation. Its absolutely continuous part on the circle (also called $w_c$) is uniformly in $A_2$.
\end{theorem}

\begin{proof}
We already saw that \eqref{A2} implies
$$
\frac{1}{1-|\phi|^2} \in L^1({\bbT})\,.
$$
Then if we put $g_c :=\frac{1-\phi_c}{\psi}$ we get that outer function $g_c$ always is such that $g_c\in H^2$.

Assume for a moment that we can prove the folowing:
\begin{equation}
\label{ic}
 \Gamma_{\frac{\bar{g_c}}{g_c}} = e^{-ic}\Gamma_{\frac{\bar{g_0}}{g_0}}\,.
 \end{equation} 
 
The last operator has norm less than $1$. In fact, $|g_0|^2$ is given to be in $A_2$ by \eqref{A2}. 
And Then by Helson-Szeg\"o's theorem we have this estimate of the norm of Hankel operator, again see\cite{Nik}.
If the Hankel operator with unimodular symbol has norm strictly lees than $1$, then the Toeplitz operator with the same symbol is  obviously  bounded from below.

We conclude that for each real $c$ Toeplitz operator $T_{\frac{\bar{g_c}}{g_c}} $ is left invertible, and, combining this with the fact that  $g_c\in H^2$, we conclude that its adjoint has a trivial kernel.
Then  $T_{\frac{\bar{g_c}}{g_c}} $ is  invertible. But unimodular symbols of invertible Toeplitz operators are $\frac{\bar{h}}{h}$ with $|h|^2\in A_2$, see \cite{Nik}.
Now $1/g_c=\frac{\psi}{1-\phi_c}$ and $1/|g_c|^2= \frac{1-|\phi_c|^2}{|1-\phi_c|^2}= \Re \frac{1+\phi_c}{1-\phi_c}$ in the unit disc. The latter function has positive real part, and so $1/|g_c|^2\in L^1(\mathbb{T})$. As $1/g_c$ is outer, we conclude that $1/g_c\in H^2$.

 Using that $1/g_c\in H^2$
and writing $\frac{\bar{g_c}}{g_c}=\frac{\bar{h}}{h}$ we conclude that $g_c= const \cdot h$. Therefore, $\frac{|1-\phi_c|^2}{1-|\phi |^2}=|g_c|^2\in A_2$.

To  finish the proof of our Theorem \ref{curios} we are left to prove \eqref{ic}. For that we want to prove 
\begin{equation}
\label{g0}
\frac{\bar{g_c}}{g_c} \in e^{-ic}\frac{\bar{g_0}}{g_0} + H^{\infty}\,.
\end{equation}
Let us write a chain of equalities:
$$
\frac{\bar{g_0}}{g_0} = \frac{\overline{1-\phi}}{1-\phi}\cdot \frac{\psi}{\bar{\psi}} = \frac{(|\psi|^2+|\phi|^2)\psi - \bar{\phi}\psi}{(1-\phi)\bar{\psi}}= -\frac{\bar{\phi}\psi}{\bar{\psi}} + \frac{\psi^2}{1-\phi}\,.
$$
Similarly,
$$
\frac{\bar{g_c}}{g_c} =e^{-ic} \frac{\overline{1-\phi_c}}{1-\phi_c}\cdot \frac{\psi_{\frac{c}2}}{\bar{\psi_{\frac{c}2}}} = -e^{-ic}\bigg[\frac{\bar{\phi_c}\psi_{\frac{c}2}}{\bar{\psi_{\frac{c}2}}} + \frac{\psi_{\frac{c}2}^2}{1-\phi_c}\bigg]=
$$
$$
-e^{-ic}\frac{\bar{\phi}\psi}{\bar{\psi}}  +\frac{\psi^2}{1-\phi_c}\,.
$$

Comparing these two equalities we obtain:
$$
\frac{\bar{g_c}}{g_c} =e^{-ic} \frac{\bar{g_0}}{g_0} + \bigg(\frac{\psi^2}{1-\phi_c}- e^{-ic}\frac{\psi^2}{1-\phi}\bigg)\,.
$$
Both functions in the brackets are from $H^{\infty}$. In fact, they are obviously from Smirnov class. 
Also their boundary values are at most $2$ by absolute value: from the above chains of equalities one can see that they are simply sums of two unimodular functions each.

\end{proof}

\bibliographystyle{amsplain}

\end{document}